\documentclass[conference]{IEEEtran}
\usepackage[left=1.45cm,right=1.45cm,top=1.9cm,bottom=2.9cm]{geometry}
\usepackage{graphicx, array, amssymb, psfrag, graphics, setspace, stfloats}
\usepackage{amsfonts}
\usepackage[toc,abbreviations,nonumberlist,nogroupskip,nopostdot]{glossaries-extra}
\glssetcategoryattribute{general}{glossdesc}{firstuc}
\glssetcategoryattribute{abbreviation}{glossdesc}{title}
\usepackage{mathtools,bm}
\usepackage{amsmath}
\usepackage{amsthm}
\usepackage{algorithm,algorithmic}

\DeclareMathOperator*{\minimize}{min.}
\newcommand{\m}[1]{\mathbf{#1}}
\newcommand{\myabs}[1]{\left\lvert#1\right\rvert}

\newcommand{\E}[1]{\mathbb{E} \left[{#1}\right]}

\def\scalingFig{0.32}

\newtheorem{theorem}{Theorem}[section]

\newglossary*{nomenclature}{Nomenclature}

\newabbreviation{ISI}{ISI}					{inter-symbol interference}
\newabbreviation{DAC}{DAC}					{digital-{}to-analog converter}
\newabbreviation{ADC}{ADC}					{analog-{}to-digital converter}
\newabbreviation{DSP}{DSP}					{digital signal processing}
\newabbreviation{TX}{TX}					{transmitter}
\newabbreviation{RX}{RX}					{receiver}
\newabbreviation{PSK}{PSK}					{phase shift keying}
\newabbreviation{QAM}{QAM}					{quadrature-amplitude modulation}
\newabbreviation{FEC}{FEC}					{forward error correction}
\newabbreviation{SOP}{SOP}					{State-{}of-polarization}
\newabbreviation{FF}{FF}					{feed-forward}
\newabbreviation{FFEQ}{FFEQ}				{feed-forward equalizer}
\newabbreviation{BER}{BER}					{bit error rate}
\newabbreviation{SNR}{SNR}					{signal-{}to-noise ratio}
\newabbreviation{RSNR}{RSNR}				{required SNR}
\newabbreviation{SNDR}{SNDR}				{Signal-{}to-noise-{}and-distortion ratio}
\newabbreviation{SFDR}{SFDR}				{Spurious free dynamic range}
\newabbreviation{RC}{RC}					{raised cosine}
\newabbreviation{RRC}{RRC}					{root raised cosine}
\newabbreviation{ENOB}{ENOB}				{Effective number of bits}
\newabbreviation{GD}{GD}					{group delay}
\newabbreviation{CMD}{CMD}					{chromatic dispersion}
\newabbreviation{PMD}{PMD}					{polarization mode dispersion}
\newabbreviation{PDL}{PDL}					{polarization dependent loss}
\newabbreviation{ASE}{ASE}					{amplified spontaneous emission}
\newabbreviation{LMS}{LMS}					{least mean squares}
\newabbreviation{DC-LMS}{DC-LMS}			{delay-compensated LMS}
\newabbreviation{APA}{APA}					{affine projection algorithm}
\newabbreviation{NLMS}{NLMS}				{normalized LMS}
\newabbreviation{MMSE}{MMSE}				{minimum mean square error}
\newabbreviation{CMA}{CMA}					{Constant modulus algorithm}
\newabbreviation{RLS}{RLS}					{recursive least squares}
\newabbreviation{LS}{LS}					{least squares}
\newabbreviation{LO}{LO}					{local-oscillator}
\newabbreviation{CR}{CR}					{carrier-recovery}
\newabbreviation{ASIC}{ASIC}				{application-specific integrated circuits}
\newabbreviation{FIR}{FIR}					{finite impulse response}
\newabbreviation{IIR}{IIR}					{infinite impulse response}
\newabbreviation{DD-LMS}{DD-LMS}			{decision-directed least mean squares}
\newabbreviation{DD}{DD}					{decision-directed}
\newabbreviation{CS-DAC}{CS-DAC}			{current-steering DAC}
\newabbreviation{LSB}{LSB}					{least-significant bit}
\newabbreviation{MSB}{MSB}					{most-significant bit}
\newabbreviation{DNL}{DNL}					{Differential non-linearity}
\newabbreviation{INL}{INL}					{Integral non-linearity}
\newabbreviation{DGD}{DGD}					{differential group delay}
\newabbreviation{FFT}{FFT}					{fast-Fourier transform}
\newabbreviation{IFFT}{IFFT}				{inverse fast-Fourier transform}
\newabbreviation{DFT}{DFT}					{discrete Fourier transform}
\newabbreviation{IDFT}{IDFT}				{inverse discrete Fourier transform}
\newabbreviation{FT}{FT}					{Fourier transform}
\newabbreviation{MSE}{MSE}					{mean square error}
\newabbreviation{HD}{HD}					{hard decision}
\newabbreviation{SD}{SD}					{soft decision}
\newabbreviation{LDPC}{LDPC}				{low-density parity check}
\newabbreviation{CW}{CW}					{continuous wave}
\newabbreviation{PBC}{PBC}					{polarization beam combiner}
\newabbreviation{MIMO}{MIMO}				{multiple-input {}and multiple-output}
\newabbreviation{SISO}{SISO}				{single-input {}and single-output}
\newabbreviation{OPGW}{OPGW}				{optical ground wire}
\newabbreviation{ZF}{ZF}					{zero-forcing}
\newabbreviation{CAZAC}{CAZAC}				{Constant amplitude zero auto-correlation}
\newabbreviation{CFO}{CFO}					{carrier frequency offset}
\newabbreviation{MA}{MA}					{moving average}
\newabbreviation{DE}{DE}					{Differential evolution}
\newabbreviation{SA}{SA}					{simulated annealing}
\newabbreviation{DEM}{DEM}					{Dynamic element matching}
\newabbreviation{LUT}{LUT}					{lookup table}
\newabbreviation{DP}{DP}					{dynamic programming}
\newabbreviation{DPC}{DPC}					{digital pre-compensation}
\newabbreviation{NN}{NN}					{neural network}
\newabbreviation{MLSE}{MLSE}				{maximum likelihood sequence estimation}
\newabbreviation{LE}{LE}					{linear equalizer}
\newabbreviation{DFE}{DFE}					{Decision–feedback equalizer}
\newabbreviation{THP}{THP}					{Tomlinson-Harashima precoding}
\newabbreviation{HW}{HW}					{hardware}
\newabbreviation{PS}{PS}					{pilot sequence}
\newabbreviation{SW-LS}{SW-LS}				{sliding window least squares}
\newabbreviation{RD-Kalman}{RD-Kalman}      {radius-directed Kalman}		
\newabbreviation{TCM}{TCM}					{trellis coded modulation}
\newabbreviation{CER}{CER}					{constellation expansion ratio}
\newabbreviation{SER}{SER}					{symbol error rate}
\newabbreviation{AWGN}{AWGN}			    {additive white Gaussian noise}
\newabbreviation{A-RLS}{A-RLS}              {approximated recursive least squares}
\newabbreviation{DC-A-RLS}{DC-A-RLS}        {delay-compensated A-RLS}
\newabbreviation{SQNR}{SQNR}                {signal-{}to-quantization-noise ratio}
\newabbreviation{TS}{TS}                    {Training Sequence}
\newabbreviation{Pol}{Pol}                  {Polarization}
\newabbreviation{RMS}{RMS}                  {root mean square}
\newabbreviation{PMF}{PMF}                  {probability mass function}
\newabbreviation{CMF}{CMF}                  {cumulative mass function}
\newabbreviation{PDF}{PDF}                  {probability distribution function}
\newabbreviation{GCS}{GCS}                  {geometric constellation shaping}
\newabbreviation{PCS}{PCS}                  {Probabilistic constellation shaping}
\newabbreviation{DM}{DM}                    {distribution matching}
\newabbreviation{SpSh}{SpSh}                {sphere shaping}
\newabbreviation{CCDM}{CCDM}                {constant composition distribution matching}
\newabbreviation{SM}{SM}                    {shell mapping}
\newabbreviation{ESS}{ESS}                  {enumerative sphere shaping}
\newabbreviation{PCDM}{PCDM}                {prefix-free code DM}
\newabbreviation{HiDM}{HiDM}                {hierarchical DM}
\newabbreviation{MPDM}{MPDM}                {multiset-partition distribution matching}
\newabbreviation{CDF}{CDF}                  {cumulative distribution function}
\newabbreviation{KL}{KL}                    {Kullback-Leibler}
\newabbreviation{AC}{AC}                    {arithmetic coding}
\newabbreviation{IID}{IID}                  {independent and identically distributed}
\newabbreviation{PAS}{PAS}                  {probabilistic amplitude shaping}
\newabbreviation{ADM}{ADM}                  {Arithmetic distribution matching}
\newabbreviation{M-ADM}{M-ADM}              {Markovian arithmetic distribution matching}

\graphicspath{{Fig/}} 
\begin{document}

\title{Joint Precoding and Probabilistic Constellation Shaping using Arithmetic Distribution Matching}
\author{
    \IEEEauthorblockN{Ramin Babaee\IEEEauthorrefmark{1}, 
    Shahab Oveis Gharan\IEEEauthorrefmark{2}, 
    and Martin Bouchard\IEEEauthorrefmark{1}} 
    \IEEEauthorblockA{\IEEEauthorrefmark{1}School of EECS, University of Ottawa, Ottawa, ON  K1N 6N5, Canada, \{ramin.babaee,bouchm\}@uottawa.ca}
    \IEEEauthorblockA{\IEEEauthorrefmark{2}Ciena Corp., Ottawa, ON K2K 0L1, Canada, soveisgh@ciena.com}
}

\maketitle
 
\begin{abstract}
The problem of joint shaping and precoding is studied in this paper. We introduce statistical dependencies among consecutive symbols to shape the constellation while minimizing the total transmit power when the signal goes through the precoding filter. We propose a stationary Markovian model for optimizing the transition probability of transmit symbols to avoid high-energy sequences when convolved with the precoding filter. A new algorithm based on arithmetic coding is proposed to generate a shaped sequence of symbols with the given Markov model transition probabilities.
\end{abstract}


\section{Introduction}

According to Shannon's famous channel capacity theorem \cite{6773024,1697831}, a continuous Gaussian distributed signal can achieve the capacity of \gls*{AWGN} channel when combined with an ideal channel coding algorithm. \gls*{PCS} has been a significant area of research in modern communication systems, aiming to optimize the distribution of constellation points to improve the overall performance of the system. \gls*{PCS} applies a non-uniform distribution to the constellation points while keeping the minimum Euclidean distance the same. The \gls*{PAS} framework proposed in \cite{7307154} provides an elegant solution for integrating shaping into existing \gls*{FEC} architectures. In \gls*{PAS}, only the amplitudes of constellation points are shaped and the signs of the symbols are determined by the uniformly distributed parity bits generated by a systematic FEC encoder. \gls*{PAS} facilitates the independent design of \gls*{FEC} and \gls*{PCS}, thereby enabling the utilization of readily available FEC codes in combination with independently optimized shaping algorithm. 

\gls*{PCS} techniques can be categorized into two main subgroups, \gls*{DM} and \gls*{SpSh}. In \gls*{DM}, the target probability distribution is \textit{directly} enforced on a low-dimensional signal by a distribution matching algorithm. In other words, a \gls*{DM} technique transforms a sequence of random bits into a sequence of symbols with a desired distribution. In \gls*{SpSh}, an $n$-dimensional sphere, where $n$ is a large number, is selected from an $n$-dimensional cube (also known as $n$-dimensional hypercube) to achieve a target rate, \textit{indirectly} inducing a Gaussian distribution on the low-dimensional signal. 

One of the notable \gls*{DM} algorithms is \gls*{CCDM} \cite{7322261}. \gls*{CCDM} \cite{7322261} operates by producing sequences that have a fixed number of occurrences for each amplitude value that satisfies a target composition. Consequently, all the generated sequences have the same energy and are basically permutations of the base sequence with a specific composition. The well-known \gls*{SpSh} techniques are \gls*{ESS} \cite{8895789,8850066,9868127,10206608} and \gls*{SM} \cite{335969,8599055}. In \gls*{ESS}, all amplitude sequences are selected that satisfy a maximum total energy constraint. The other approach to sphere shaping is to sort the sequences based on their energies. Sequences within the same shell energy can be sorted lexicographically. In \gls*{SM}, a divide-and-conquer approach is utilized to find the mapping between indices and amplitude sequences. This is achieved by breaking the $n$-dimensional indexing problem into two $n/2$-dimensional mapping. The total number of steps is $\lceil \log_2{n} \rceil$.

In addition to shaping, equalization also plays a crucial role in communication systems. The main drawback of linear equalization at the receiver is noise enhancement in severely distorted channels. This causes a significant increase in the noise power at the equalizer output, particularly when the channel has zeros close to the unit circle \cite{proakis}. \gls*{DFE} \cite{1455678} combats this issue by using previously detected symbols and canceling the \gls*{ISI}. The main drawback is the propagation of errors when a symbol is detected incorrectly. To mitigate this problem, if the channel state information is available at the transmitter, the feedback equalizer can be designed and implemented at the transmit end. However, this results in an increase in transmit signal power. 

\gls*{THP} \cite{Tomlinson1971NewAE, 1091221} applies a non-linear modulo operation in order to bound the amplitude of the transmitted signal, which overcomes the increase in transmit power. Although \gls*{THP} has been initially proposed for ISI cancellation at the transmitter for \gls*{SISO} channels, it has been extended to spatial equalization in \gls*{MIMO} systems where inter-user interference cancellation is performed \cite{991747,1415910,1494693,6797864}. An advantage of \gls*{THP} is that coded modulation techniques such as \gls*{TCM} can be utilized along with \gls*{THP} to achieve close to channel capacity performance \cite{120349}. Although \gls*{THP} can also be utilized in conjunction with recent coding algorithms such as \gls*{LDPC} codes, it cannot be used with constellation shaping methods in a straightforward manner. Signal characteristics at the input of the precoder can be completely modified by the non-linear modulo operator, resulting in a non-optimal signal distribution at the channel input. 
 
The proposed \gls*{PCS} algorithms in the literature are designed without considering linear equalization at the transmitter. However, the applicability of PCS approaches may be very limited in scenarios where linear pre-equalization of the channel is employed at the transmitter and when there are constraints on the transmit signal power. On the other hand, existing precoding algorithms such as \gls*{THP} cannot be used directly with the shaping algorithms. 



In this paper, we focus on the problem of joint precoding and probabilistic constellation shaping. We introduce statistical dependencies between consecutive symbols in order to minimize the transmit power. This is accomplished by formulating the shaping problem as a Markovian process. We show that the problem is convex and can be solved numerically. A new arithmetic distribution matching scheme is then presented for transforming binary data to sequences with desired conditional probabilities. The proposed algorithm can be employed in many applications such as coherent optical communications where the transmitter equalization response is known and static. The integration of the proposed algorithm with FEC is beyond the scope of this paper and is an interesting topic for future research.

The rest of the paper is organized as follows. Section \ref{sec:prob_precoding} discusses the proposed shaping scheme. We propose \gls*{M-ADM} algorithm in Section \ref{sec:madm} that is used to shape the constellation points for a Markov model. The simulation results are provided in Section \ref{sec:simulaion_results} followed by the conclusions in Section \ref{sec:conclusion}.

\textbf{Notation:} Random variables are denoted by italic capital letters $X$ and realizations by italic small letters $x$. Vectors and matrices are represented by bold non-italic letters $\m{X}$. $(\cdot)^\text{T}$ indicates transpose of a vector/matrix. Symbol $\E{\cdot}$ is used for statistical expectation operation. Sets are represented by calligraphic font $\mathcal{R}$ and the cardinality of set $\mathcal{R}$ is denoted by $|\mathcal{R}|$.

\section{Joint Shaping and Precoding}
\label{sec:prob_precoding}

\begin{figure}[!tb]
\footnotesize
\includegraphics[scale=0.6]{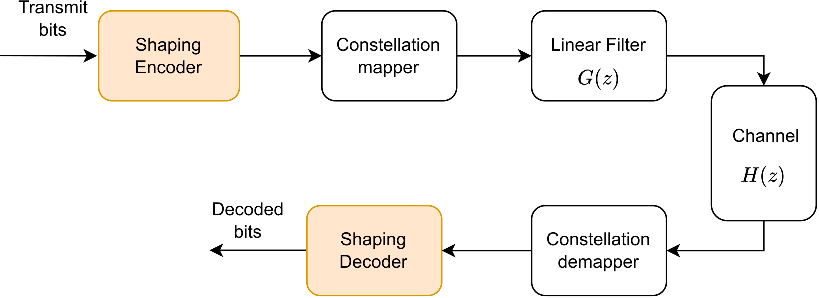}
\caption{Block diagram of transmitter, channel, and receiver for the proposed joint shaping and precoding algorithm.}
\label{fig_thp:shaping_enc_dec}
\end{figure}

In this section, we derive a probabilistic constellation shaping algorithm for joint shaping and pre-equalization. Let $g_m$ represent the precoding filter applied at the transmitter with $z$-transform $G(z) = \frac{1}{H(z)}$ where $H(z)$ is the $z$-transform of the channel. We assume $g_m$ is an $L$-tap \gls*{FIR} filter. The incoming bits are mapped to symbols $a_m$ which are drawn from an $M_B$-ASK constellation set $\mathcal{C}=\{\pm1, \pm3, \pm{(M_B-1)}\}$, where $M_B$ is an even integer. Let $\m{A}(a_{m-L+1}, \dots, a_{m}) = [a_{m-L+1}, \dots, a_{m}]$ denote a vector of $L$ consecutive transmit symbols. The precoded signal power can be formulated as
\begin{equation}
\begin{aligned}
    \sigma^2_T & = \E{\myabs{\m{g}^\text{T} \m{A}(a_{m-L+1}, \dots, a_{m})}^2} \\
               & = \sum_{a_{m-L+1}} \ldots \sum_{a_{m}} \Big( P(a_{m-L+1}, \dots, a_{m}). \\
               & \hspace{15ex} \myabs{\m{g}^\text{T} \m{A}(a_{m-L+1}, \dots, a_{m})}^2 \Big), \\
\end{aligned}
\end{equation}
where $\m{g} = \big[g_{L-1}, \dots, g_1, g_0\big]^\text{T}$ and $P(a_{m-L+1}, \dots, a_{m})$ is the \gls*{PMF} of sequence $\m{A}(a_{m-L+1}, \dots, a_{m})$ which we assume is being generated by our \gls*{PCS} encoder engine. The design objective is to optimize the sequence probabilities such that the information rate of the system is at least $R$ bits/sym and transmit power is minimized. The block diagram of the transmitter and receiver is illustrated in Fig. \ref{fig_thp:shaping_enc_dec}.

Note that when $g_m$ is the Kronecker delta function $\delta_n$ and therefore no filtering is applied at the transmitter, this optimization problem translates into the well-known probabilistic constellation shaping, where the solution is given by the Maxwell-Boltzmann distribution $P_A(a) = K (\lambda) e^{-\lambda a^2}$, where $\lambda$ determines the information rate and $K(\lambda)$ is the normalization factor.


In our proposed method, PCS is achieved by an encoder engine which introduces dependency between consecutive transmitted symbol by means of a Markov model. In other words, output symbols are generated through a Markov model of order $L-1$ with optimized transition probabilities. We first assume that $L=2$, i.e., a $2$-tap transmit equalizer, and derive the optimization problem and then generalize it for an arbitrary $L$. The information rate of a Markov model can be calculated from its conditional entropy as \cite[Chapter 4.2]{cover_book}
\begin{equation}
\mathbb{H}(a_{m}|a_{m-1}) = - \sum_{a_{m-1}} \sum_{a_{m}} P(a_{m-1},a_{m}) \log{P(a_{m}|a_{m-1})},
\end{equation}
where $P(a_{m}|a_{m-1})$ indicates probability of sending symbol $a_m$ conditioned on symbol $a_{m-1}$ was being sent in previous time-slot. Here, our aim is to minimize the signal power at the transmitter output subject to the constraint on the information rate of the \gls{PCS} encoder engine. Hence, we can formulate the problem as the following optimization problem
\begin{equation}
\label{eq_pcsh:form1}
\begin{aligned}
     & {\minimize_{P(a_{m-1},a_{m})}} && \E{ \myabs{\m{g}^\text{T} \m{A}(a_{m-1},a_{m})}^2 } \\
     & \text{\hspace{2ex} s.t.} && \mathbb{H}(a_{m}|a_{m-1}) \ge R \\
     & && \pi_{a_{m}} = \pi_{a_{m-1}}, \quad \forall a_{m-1}=a_{m} \in \mathcal{C} 
\end{aligned}
\end{equation}
where $R$ denotes the rate constraint on the PCS encoder engine, and $\pi_{a_m}$ and $\pi_{a_{m-1}}$ are the marginal \gls*{PMF}s of symbol $a_m$ and $a_{m-1}$, respectively and can be obtained as
\begin{align}
   & \pi_{a_m}      = \sum_{a_{m-1}}P(a_{m-1},a_m), \\ 
   & \pi_{a_{m-1}}  = \sum_{a_{m}}P(a_{m-1},a_m).
\end{align}
The proposed shaping algorithm forms a Markov chain in which the probability distribution of symbol $a_m$ depends on the value of symbol $a_{m-1}$. In order to ensure that the probability distribution remains unchanged from $a_{m-1}$ to $a_m$, we need a stationary distribution for the Markov process. This requirement is satisfied by the second constraint of (\ref{eq_pcsh:form1}). The optimization problem (\ref{eq_pcsh:form1}) can be formulated as a function of $P(a_{m-1},a_{m})$ as
\begin{equation}
\begin{aligned}
         & \minimize_{P(a_{m-1},a_{m})} \; && \sum_{a_{m-1}} \sum_{a_{m}} P(a_{m-1},a_{m}) \myabs{\m{g}^\text{T} \m{A}(a_{m-1},a_{m})}^2 \\
         & \hspace{4ex} \text{s.t.  } && \mathbb{H}(a_{m}|a_{m-1}) = -\sum_{a_{m-1}} \sum_{a_{m}} \Big( P(a_{m-1},a_{m}) \\ & && \hspace{6ex} \log{\frac{P(a_{m-1},a_{m})}{\sum_{a_{m}}P(a_{m-1},a_{m})}} \Big) \ge R,\\
         & && \sum_{a_{m-1}}P(a_{m-1},s) = \sum_{a_{m}}P(s,a_{m}), \quad \forall s \in \mathcal{C} \\
         & && \sum_{a_{m-1}} \sum_{a_{m}}P(a_{m-1},a_{m}) = 1, \\
         & && P(a_{m-1},a_{m}) \ge 0, \quad \forall a_{m-1}, a_{m} \in \mathcal{C}.\\
\end{aligned}
\end{equation}

The generalization for an arbitrary $L$-tap transmitter equalizer is as follows
\begin{equation}
\label{eq_pcsh:form2}
\begin{aligned}
         & \minimize_{P} && \sum_{a_{m-L+1}} \ldots \sum_{a_{m}} P(a_{m-L+1}, \dots, a_{m}) \\ & && \hspace{18ex} \myabs{\m{g}^\text{T} \m{A}(a_{m-L+1}, \dots, a_{m})}^2 \\
         & \hspace{1ex} \text{s.t.} && \mathbb{H}(a_{m}|a_{m-1},\dots,a_{m-L+1}) = \\
         & && -\sum_{a_{m-L+1}} \ldots \sum_{a_{m}} \Big( P(a_{m-L+1}, \dots, a_{m}) \\ & && \hspace{8ex} \log{\frac{P(a_{m-L+1}, \dots, a_{m})}{\sum_{a_{m}}P(a_{m-L+1}, \dots, a_{m})}} \Big) \ge R,\\
         & && \sum_{a_{m-L+1}}P(a_{m-L+1},s_1\ldots, s_{L-1}) = \\ & && \hspace{2ex}  \sum_{a_{m}}P(s_1,\ldots, s_{L-1}, a_{m}), \quad \forall s_1,\ldots, s_{L-1} \in \mathcal{C} \\
         & && \sum_{a_{m-L+1}} \ldots \sum_{a_{m}} P(a_{m-L+1}, \dots, a_{m}) = 1, \\
         & && P(a_{m-L+1}, \dots, a_{m}) \ge 0, \quad \forall a_{m-L+1}, \dots, a_{m} \in \mathcal{C}.\\
\end{aligned}
\end{equation}
Note that the objective function and all constraints except conditional entropy are affine functions of $P(a_{m-L+1}, \dots, a_{m})$. The concavity of the conditional entropy function is proved in Appendix \ref{AppendixCondEntr}. Therefore, the optimization problem is convex and can be efficiently calculated using the interior-point method \cite{boyd2004convex}. 

\section{Markovian Arithmetic Distribution Matching}
\label{sec:madm}
\gls*{ADM} algorithms proposed in the literature assume that the desired encoded sequence is \gls*{IID}; therefore, each output symbol is drawn independently from a fixed marginal distribution. We propose \gls*{M-ADM} that is suitable for target distributions with transition probabilities. While we present the algorithm for a first-order Markov chain, the extension to higher orders follows similar principles.

\gls*{ADM} transforms Bernoulli($\frac{1}{2}$) distributed binary sequence $\bm{s}$ of length $k$ into the codeword $\bm{c}$ of length $n$ with the desired marginal distribution. The codeword symbols are drawn from the output alphabet $\mathcal{C} = \{c_0, c_1, ..., c_{n-1}\}$. Each unique input sequence $\bm{s}_i$, $i=0,...,2^k-1$ corresponds to a specific subinterval within the interval $[0, 1)$. Similarly, each codeword $\bm{c}_i$, $i=0,...,2^k-1$ corresponds to a distinct subinterval $I(\bm{c}_i)$ within $[0, 1)$. These non-overlapping subintervals form a complete partition of $[0, 1)$. An input sequence $\bm{s}$ is encoded by finding the subinterval $I(\bm{c})$ which is a subset of the input subinterval. The associated codeword $\bm{c}$ is the output. The received codeword $\bm{c}$ identifies the corresponding subinterval $I(\bm{c})$ and can be correctly decoded back to the original input sequence $\bm{s}$ that contains $I(\bm{c})$.

For the purpose of clarity and consistency, we adopt the same terminologies as those used in \cite{7052132,7322261} to describe the basic operations of the proposed encoder and decoder. Source interval and code interval denote the probability ranges for the symbols on the source side and the encoded output side, respectively. An interval $[a, b)$ is said to identify an interval $[c, d)$ if $[a, b) \subseteq [c, d)$. Interval refinement is defined as the partitioning of the interval based on the \gls*{CMF}, followed by the selection of one of the resulting subintervals. We assume the source has a Bernoulli($\frac{1}{2}$) distribution $p_\text{source}$ and the conditional distribution of the desired output is defined as $p_\text{code}(c_i|c_{i-1})$. The refinement of the source interval is based on the independent probability distribution $p_\text{source}$, while code interval refinement depends on the conditional probability $p_\text{code}(c_i|c_{i-1})$, which requires knowledge of the previous symbol. 



The \gls*{M-ADM} encoding process begins by initializing the source and code intervals to $[0, 1)$. In a loop, the encoder reads a new input bit. With a new input bit, the source interval is recursively refined on the basis of the source \gls*{PMF}. That means that the subintervals $[0,0.5)$ and $[0.5,1)$ are assigned to input 0 and 1, respectively. If the first bit is 0, the lower subinterval of the source interval is selected; otherwise the upper subinterval is chosen as the new source interval. After reading the input, if the refined source interval identifies a candidate in the code list (based on $p_\text{code}(c_i|c_{i-1})$), the corresponding code symbol is appended to the output and the code interval is updated accordingly. This procedure continues until the source interval is not identifiable by any of the candidate intervals. The refinement and output assignment process continues as long as there are overlaps. After all input bits are read, the encoder finds two candidates that identify the source interval. The larger of the two intervals is chosen, and the corresponding symbols are outputted. The pseudocode of the encoder is described in Algorithm \ref{alg:enc}.

\begin{algorithm}[t]
\caption{\gls*{M-ADM} encoder}
\begin{algorithmic}[1] 
\label{alg:enc}
\STATE \textbf{Input:} 
    \STATE \hspace{1ex} Input sequence:         $\bm{s}$
    \STATE \hspace{1ex} Source distribution: $p_{\text{source}}$
    \STATE \hspace{1ex} Conditional target distribution: $p_{\text{code}}(c_i | c_{i-1})$  
\STATE \textbf{Output:}
    \STATE \hspace{1ex}  Encoded sequence: $\bm{c}$ 

\STATE \textbf{Initialization:} 
    \STATE $\bm{c} \gets []$ \COMMENT{Output buffer}
    \STATE $I_{\text{source}} \gets [0, 1)$ \COMMENT{Initial source interval}
    \STATE $I_{\text{code}} \gets [0, 1)$  \COMMENT{Initial code interval}
    \STATE $c_\text{last} \gets c_0$ \COMMENT{Last encoded symbol} 

\FOR {$k = 1$ \TO $\text{length}(\bm{s})$} 
\STATE $s_\text{new} \gets \bm{s}[k]$  \COMMENT{Read new input bit}

\STATE $I_{\text{source}} \gets \text{RefineInterval}(I_{\text{source}}, s_\text{new}, p_{\text{source}})$

\WHILE {$I_{\text{source}}$ identifies a code interval $j$}
    \STATE $\text{Append } \text{code}[j] \text{ to } \bm{c}$  
    \STATE $I_{\text{code}} \gets \text{RefineInterval}(I_{\text{code}},\text{code}[j],p_\text{code}(c_i | c_\text{last}))$ 
    \STATE $c_\text{last} \gets \text{code}[j]$     
\ENDWHILE
\ENDFOR

\STATE \textbf{Finalization:}
\STATE $I^u_{\text{code}} \gets \text{RefineInterval}(I_{\text{code}}, p_{\text{code}})$ \COMMENT{Until $I^u_{\text{code}}$ identifies $I_{\text{source}}$} 
\STATE $I^l_{\text{code}} \gets \text{RefineInterval}(I_{\text{code}}, p_{\text{code}})$ \COMMENT{Until $I^l_{\text{code}}$ identifies $I_{\text{source}}$}
\STATE $I_{\text{code}} \gets \max(I^u_{\text{code}}, I^l_{\text{code}})$ 
\STATE $\text{Append corresponding symbols to } \bm{c}$

\end{algorithmic}
\end{algorithm}

The arithmetic decoder operates iteratively to reconstruct the original bit sequence. The decoder reads the symbols from the encoded data stream $\bm{c}$ until $I_\text{code}$ identifies a unique source interval. The corresponding bit is appended to the output buffer $\bm{s}$, and this selected candidate becomes the new $I_\text{source}$. This refinement and selection process continues until , In the absence of a match, the decoder resumes reading symbols from input. The decoding process is terminated when the $n$ bits, representing the original source sequence length, are written to the output buffer. Algorithm \ref{alg:dec} presents a pseudocode for the decoder. 

Note that since the proposed \gls*{M-ADM} generates variable-length outputs, buffers are required to average the rate. The proposed algorithm is referred to as probabilistic precoding for the rest of this paper.

\begin{algorithm}[t]
\caption{\gls*{M-ADM} decoder}
\begin{algorithmic}[1]
\label{alg:dec}
\STATE \textbf{Input:} 
    \STATE \hspace{1ex} Encoded sequence: $\bm{c}$
    \STATE \hspace{1ex} Source distribution: $p_{\text{source}}$
    \STATE \hspace{1ex} Conditional target distribution: $p_{\text{code}}(c_i | c_{i-1})$  
    \STATE \hspace{1ex} length of output sequence $\bm{s}$: $n$
\STATE \textbf{Output:}
    \STATE \hspace{1ex} Source sequence: $\bm{s}$ 

\STATE \textbf{Initialization:} 
    \STATE $\bm{s} \gets []$ \COMMENT{Output buffer}
    \STATE $I_{\text{source}} \gets [0, 1)$ \COMMENT{Initial source interval}
    \STATE $I_{\text{code}} \gets [0, 1)$  \COMMENT{Initial code interval}
    \STATE $c_\text{last} \gets c_0$ \COMMENT{Last input symbol}
    \STATE $k \gets 0$
\WHILE { length$(\bm{s}) < n$}
    \STATE $c_\text{new} \gets \bm{c}[k]$  \COMMENT{Read new input symbol}
    \STATE $I_{\text{code}} \gets \text{RefineInterval}(I_{\text{code}}, c_\text{new}, p_{\text{code}}(c_i | c_\text{last}))$    
    \STATE $c_\text{last} \gets \bm{c}[k]$         
    \STATE $k \gets k+1$
\WHILE {$I_{\text{code}}$ identifies a source interval $j$}
    \STATE $\text{Append} \text{ bit $j$ to } \bm{s}$  
    \STATE $I_{\text{source}} \gets \text{RefineInterval}(I_{\text{source}},j, p_\text{source})$ 
\ENDWHILE

\ENDWHILE

\end{algorithmic}
\end{algorithm}

\section{Simulation Results}
\label{sec:simulaion_results}

In order to assess the efficiency of the proposed algorithm, numerical results are obtained by running Monte Carlo simulations.  

We consider a second-order \gls*{IIR} channel with transfer function $H(z) = 1/(1+cz^{-1} + d z^{-2})$ where the precoding filter is $G(z) = 1+cz^{-1} + d z^{-2}$. The performance of the proposed probabilistic precoding scheme is investigated in Fig. \ref{fig_thp:pcs_gain_2ASK}, \ref{fig_thp:pcs_gain_4ASK}, and \ref{fig_thp:pcs_gain_8ASK} for a transmission rate of 1, 2, and 3 bits/sym, respectively. The $x$-axis is the parameter $c$ of the channel model, and the results for several values of the parameter $d$ are plotted. The $y$-axis is the shaping gain in dB, defined as the reduction in average power of the transmit signal relative to the \gls*{THP} method. The THP is used with unshaped constellation 2-ASK, 4-ASK, and 8-ASK corresponding to $R=1, 2,$ and $3$. The base constellation for the probabilistic precoding is 4-ASK, 8-ASK, and 16-ASK, respectively. The results are obtained by simulating a transmission over $10^7$ symbols and calculating the transmit signal power for each scheme. For probabilistic precoding, the optimum power is computed from the optimization problem (\ref{eq_pcsh:form2}). For comparison, the results for traditionally shaped constellation with linear pre-equalization and THP are also plotted. It is clear from the figures that using shaped constellation with either linear filtering or THP has a worse performance than the THP with uniform signaling in most channel conditions. We observe that a significant shaping gain is attained relative to THP for the proposed probabilistic precoding scheme. As the transmission rate increases, the shaping gain approaches the ultimate shaping gain of 1.53 dB. 

\begin{figure}[!tb]
\footnotesize
\psfrag{c}{$c$}
\includegraphics[scale=\scalingFig]{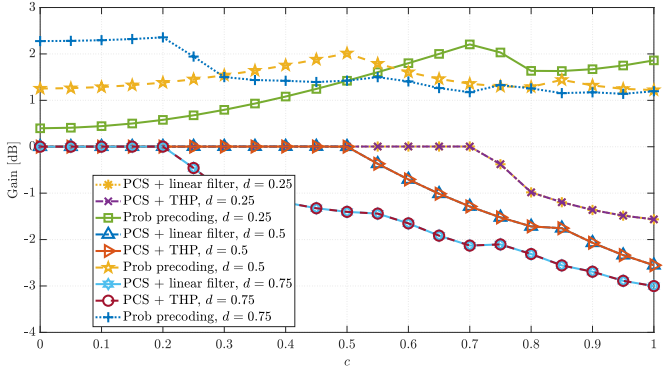}
\caption{Performance gain comparison of probabilistic precoding for transmission rate $R=1$ bit/sym over IIR channel $H(z) = 1/(1+cz^{-1} + d z^{-2})$ versus THP. The base constellation for THP is 2-ASK and for probabilistic precoding is 4-ASK.}
\label{fig_thp:pcs_gain_2ASK}
\end{figure}

\begin{figure}[!tb]
\footnotesize
\includegraphics[scale=\scalingFig]{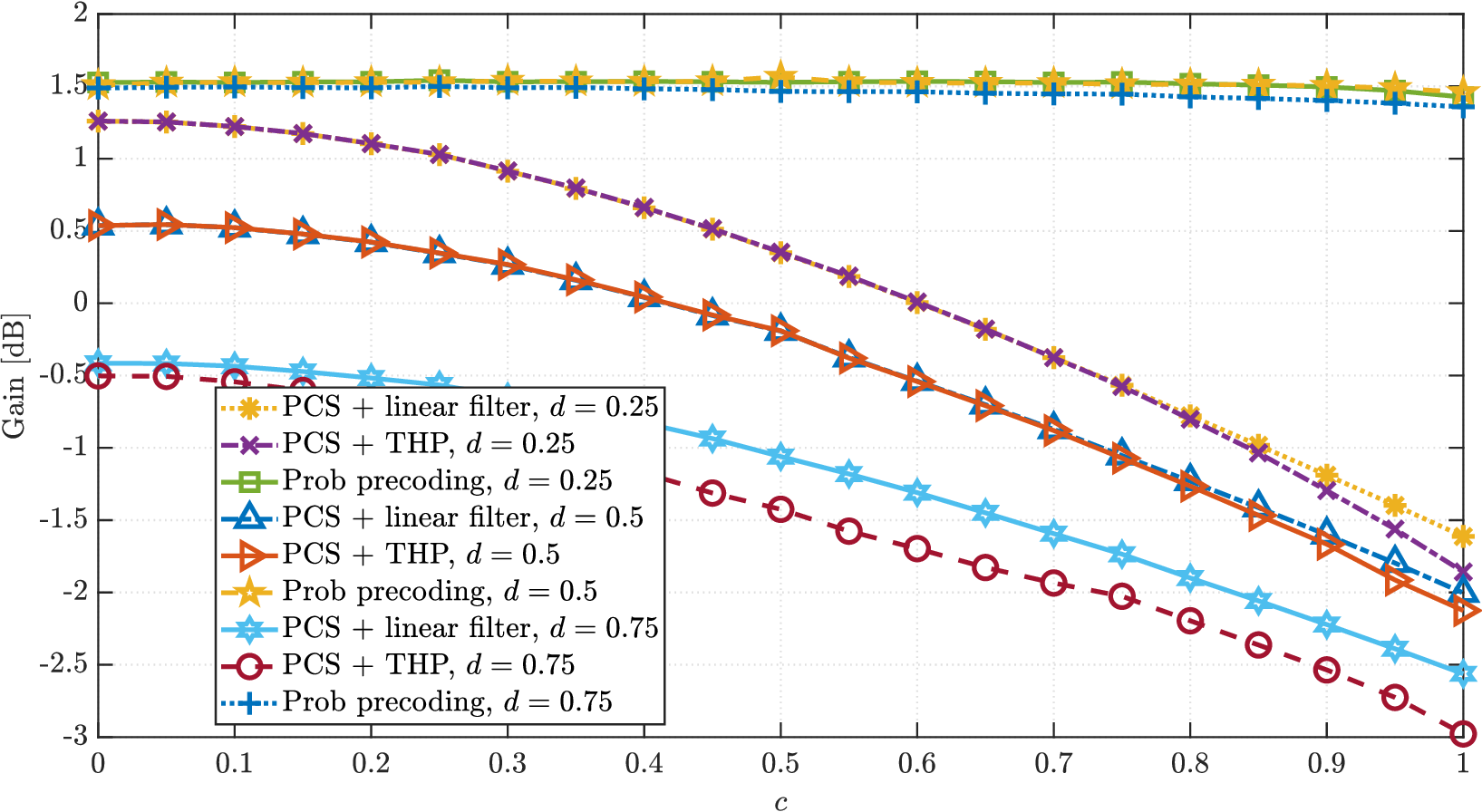}
\caption{Performance gain comparison of probabilistic precoding for transmission rate $R=2$ bit/sym over IIR channel $H(z) = 1/(1+cz^{-1} + d z^{-2})$ versus THP. The base constellation for THP is 4-ASK and for probabilistic precoding is 8-ASK.}
\label{fig_thp:pcs_gain_4ASK}
\end{figure}

\begin{figure}[!tb]
\footnotesize
\includegraphics[scale=\scalingFig]{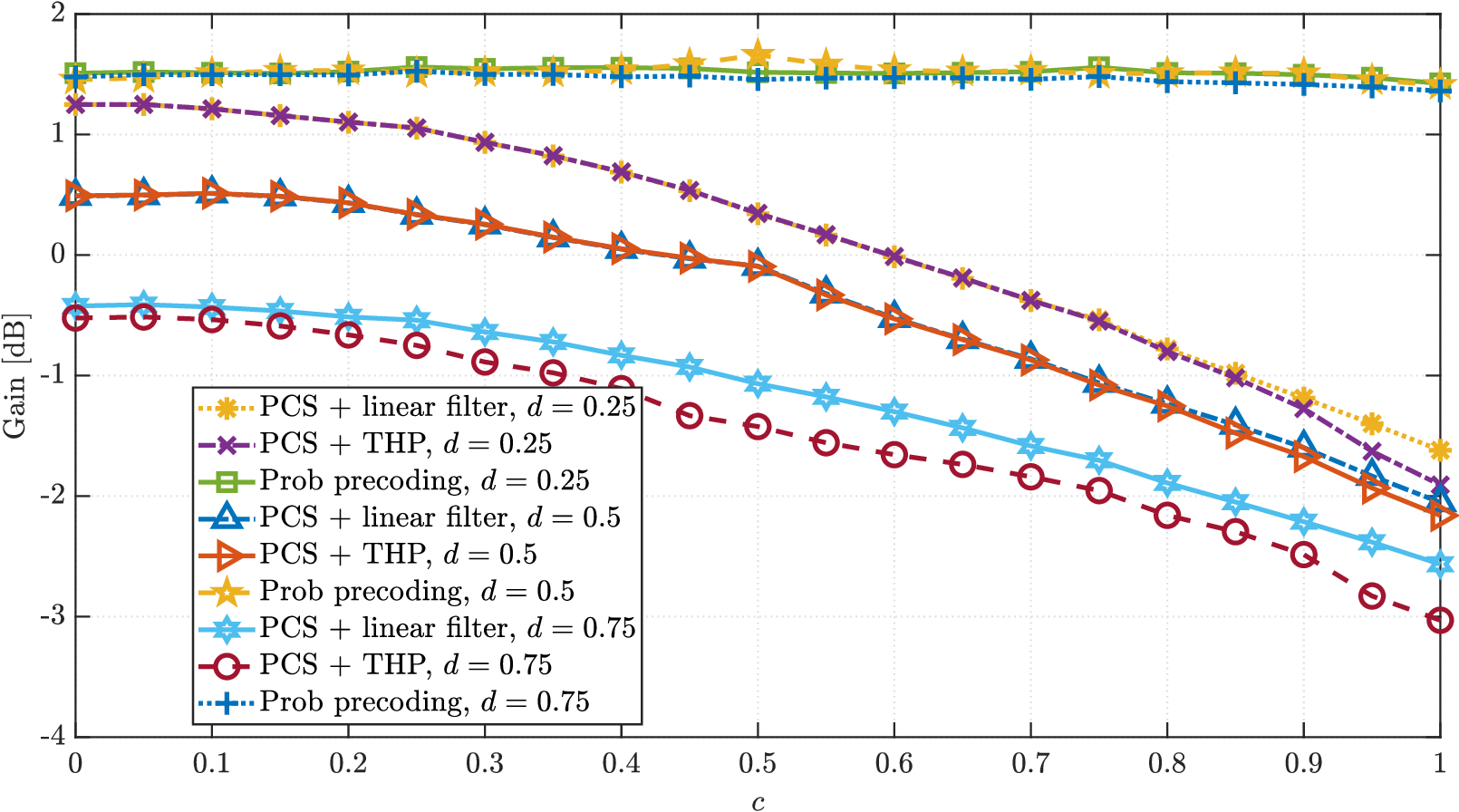}
\caption{Performance gain comparison of probabilistic precoding for transmission rate $R=3$ bit/sym over IIR channel $H(z) = 1/(1+cz^{-1} + d z^{-2})$ versus THP. The base constellation for THP is 8-ASK and for probabilistic precoding is 16-ASK.}
\label{fig_thp:pcs_gain_8ASK}
\end{figure}

Fig. \ref{fig_thp:gain_4ASK_vs_M} illustrates the impact of the size of base constellation $M_B$ on the performance of probabilistic precoding for a transmission rate of 2 bits/sym. The channel model is a first-order IIR filter $H(z) = 1/(1+cz^{-1})$. As evident in the figure, increasing the size of the shaped constellation leads to a significant improvement in heavily distorted channels. 

\begin{figure}[!tb]
\footnotesize
\includegraphics[scale=\scalingFig]{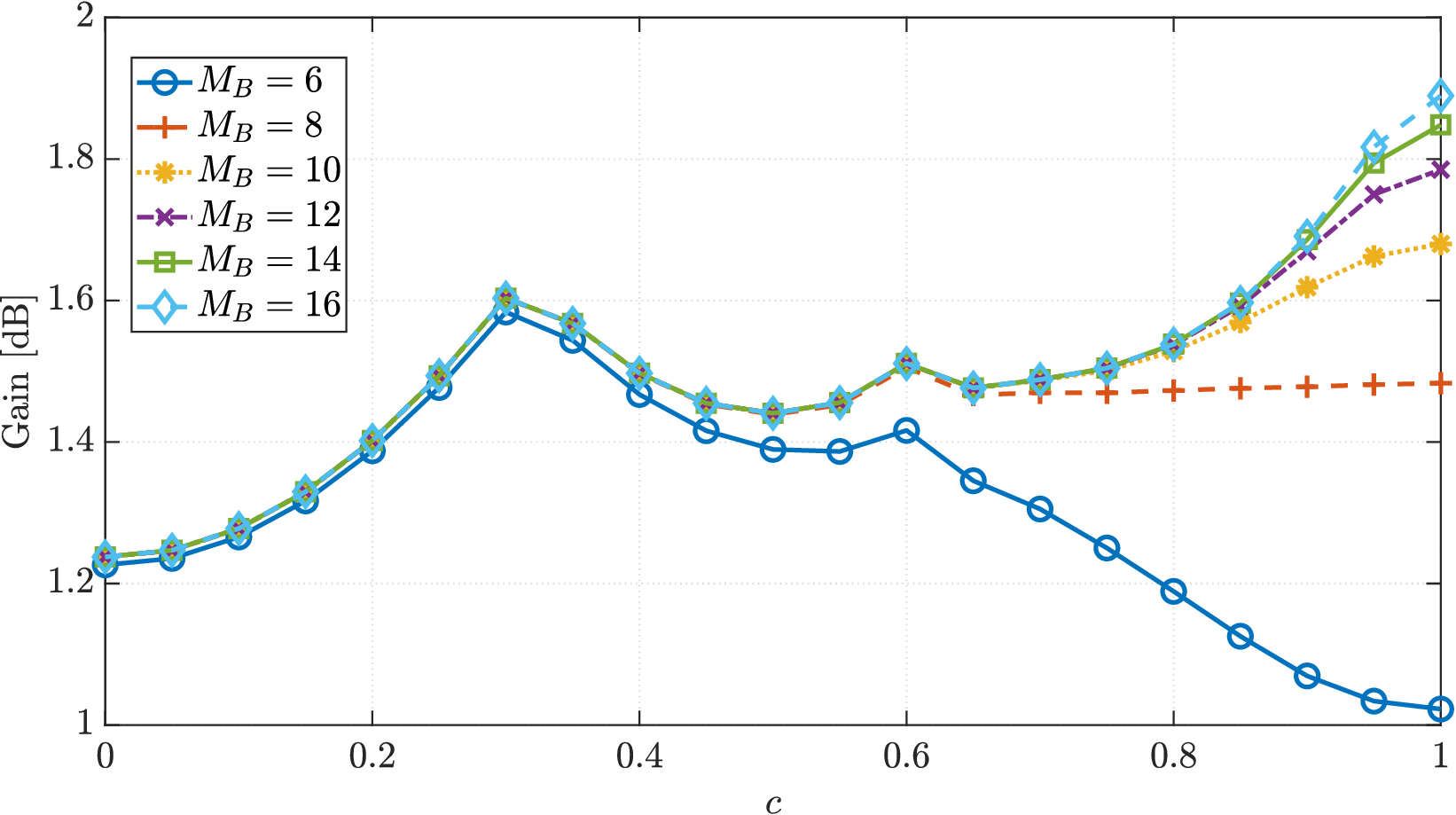}
\caption{Theoretical shaping gain of probabilistic precoding compared to THP for transmission rate $R=2$ bits/sym over IIR channel $H(z) = 1/(1+cz^{-1})$ for several values of $M_B$. The base constellation is $M_B$-ASK.}
\label{fig_thp:gain_4ASK_vs_M}
\end{figure}

Fig. \ref{fig_thp:gain_4ASK_AC_SM} compares the performance of \gls*{M-ADM} with the theoretical result. The channel model is $H(z) = 1/(1+cz^{-1})$ and the transmission rate is $R = 2$ bits/sym. For \gls*{M-ADM}, the length of the bit sequence of the encoder is 512 bits, corresponding to a block size of 256 symbols at the output of the encoder. Note that \gls*{M-ADM} produces variable-length output sequences, and therefore buffering is required to achieve the average sequence length of 256 symbols. As is evident in the figure, \gls*{M-ADM} operates close to the theoretical result.

\begin{figure}[!tb]
\footnotesize
\includegraphics[scale=\scalingFig]{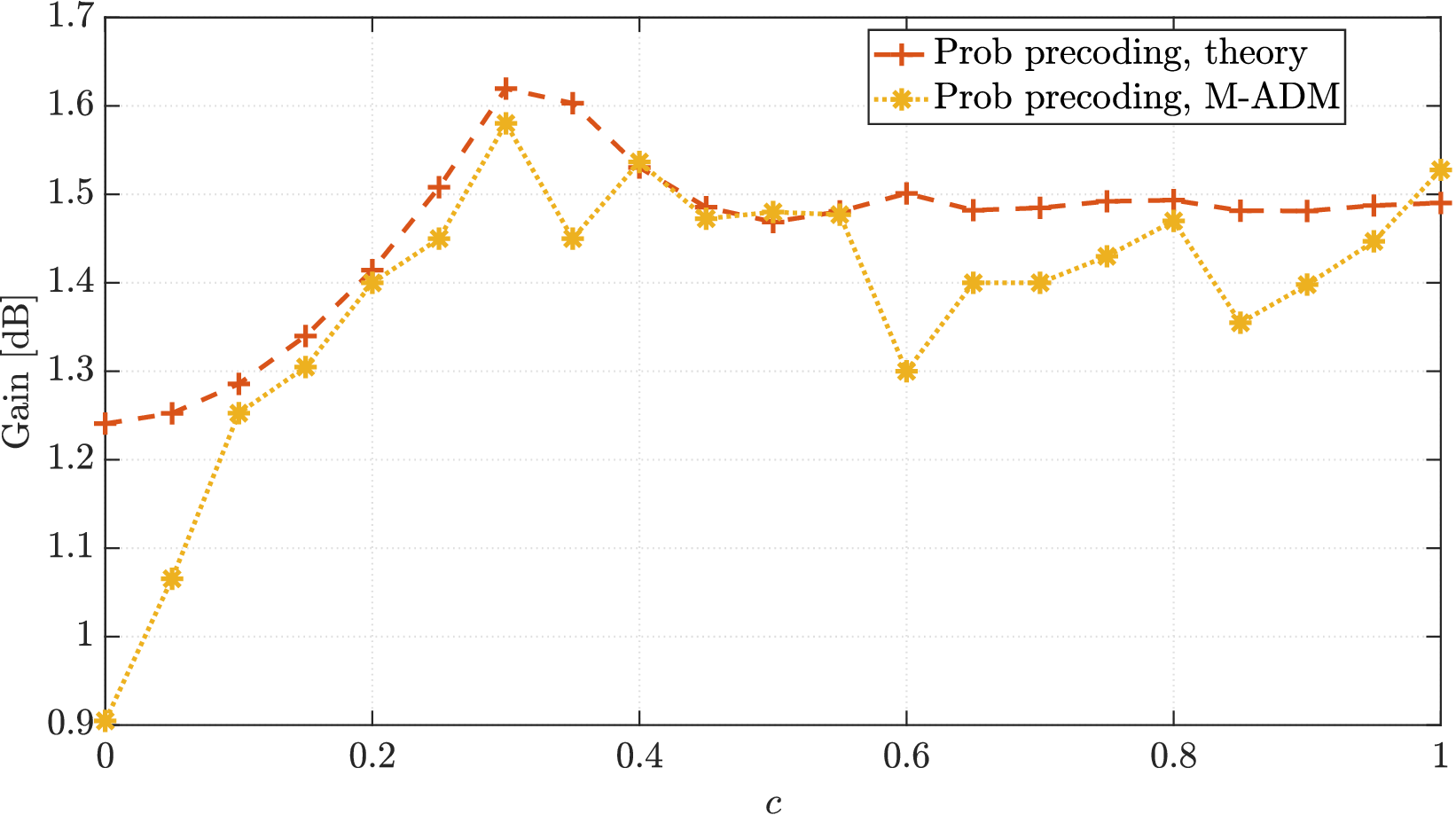}
\caption{Performance of the shaping gain achieved with \gls*{M-ADM} compared to the theoretical result. The transmission rate is $2$ bits/sym and the channel model is $H(z) = 1/(1+cz^{-1})$.}
\label{fig_thp:gain_4ASK_AC_SM}
\end{figure}

\section{Conclusion}
\label{sec:conclusion}

Under the assumption of transmit power constraint and linear pre-equalization at the transmitter, we introduced a stationary Markovian model in which the transition probability of transmit symbols is optimized. We then proposed a new distribution matching algorithm that can be used to implement the proposed Markov chain. The simulation results show promising shaping gains compared to the THP. The \gls*{M-ADM} scheme is a fixed-to-variable length transformer and consequently buffering and synchronization are needed for implementation in practical applications. As the next steps of this research, we will focus on finding practical mapping schemes where input and output codewords are fixed in length. The problem of integration of the proposed shaping algorithm with FEC is also an interesting topic to investigate. 

\appendices
\section{}
\label{AppendixCondEntr}
\begin{theorem}
The conditional entropy $H(Y|X)$ defined as
\begin{equation}
    H_P(Y|X) = -\sum_{x} \sum_{y} P_{XY}(x,y) \cdot \log{P(y|x)}
\end{equation}
is concave in joint probability mass function $P_{XY}$.
\end{theorem}

\begin{proof}
Let $X$ and $Y$ be two discrete random variables with possible outcomes $\mathcal{X}$ and $\mathcal{Y}$ and marginal \gls*{PMF} $P_X(x)$ and $P_Y(y)$, respectively. Let $P_{XY}(x,y)$ denote the joint \gls*{PMF} of $X$ and $Y$, where $x \in \mathcal{X}$ and $y \in \mathcal{Y}$. The objective is to prove that
\begin{equation}
    H_{\lambda P_1 + (1-\lambda) P_2}(Y|X) \ge \lambda H_{P_1}(Y|X) + (1-\lambda) H_{P_2}(Y|X)
\end{equation}

The conditional entropy can be written as
\begin{equation}
\begin{aligned}
H_P(Y|X) & = -\sum_{x} \sum_{y} P_{XY}(x,y) \cdot \log{\big(\frac{P_{XY}(x,y)}{\sum_{y}P_{XY}(x,y)}\big)} \\
         & = -\sum_{x} \sum_{y} P_{XY}(x,y) \cdot \log{\big(\frac{P_{XY}(x,y)}{P_X(x)}\big)} \\
         & = -\sum_{x} \sum_{y} P_{XY}(x,y) \cdot \log{\big(\frac{P_{XY}(x,y)}{P_X(x) U_Y(y)}U_Y(y)\big)} \\
         & = -\sum_{x} \sum_{y} P_{XY}(x,y) \cdot \log{\big(\frac{P_{XY}(x,y)}{Q_{XY}(x,y)}U_Y(y)\big)} \\
\end{aligned}
\end{equation}
where $U_Y(y)$ is the PMF of a discrete uniform distribution on $Y \in \mathcal{Y}$ and $Q_{XY}(x,y) = P_X(x) U_Y(y)$. Note that we have
\begin{equation}
    U_Y(y) = \frac{1}{{\myabs{\mathcal{Y}}}}.
\end{equation}
We rewrite
\begin{equation}
\label{eq_appendix:entrop1}
\begin{aligned}
H_P(Y|X) & = -\sum_{x} \sum_{y} P_{XY}(x,y) \log{\big(\frac{P_{XY}(x,y)}{Q_{XY}(x,y)}\big)}  - \\ & \hspace{8ex} \sum_{x} \sum_{y} P_{XY}(x,y) \log{\frac{1}{{\myabs{\mathcal{Y}}}}}\\
        & = -\mathbb{D}(P \| Q)+\myabs{\mathcal{Y}}
\end{aligned}
\end{equation}
where $\mathbb{D}(P \| Q)$ is the informational divergence, also known as the \gls*{KL} divergence of two discrete distributions $P(x)$ and $Q(x)$ and is defined as
\begin{equation}
\mathbb{D}(P \| Q) = \sum_{x \in \mathcal{X}}  P(x) \log{\frac{P(x)}{Q(x)}}.
\end{equation}
Given the convexity of \gls*{KL} divergence in the pair of probability distributions $P$ and $Q$ \cite{boyd2004convex}, we have
\begin{equation}
\label{eq_appendix:entrop2}
\begin{aligned}
& \mathbb{D}(\lambda P_1 + (1-\lambda) P_2 \| \lambda Q_1 + (1-\lambda) Q_2)  \le \\ & \hspace{18ex} \lambda \mathbb{D}(P_1 \| Q_1) +  (1-\lambda) \mathbb{D}(P_2 \| Q_2)
\end{aligned}
\end{equation}
By substituting Eq. (\ref{eq_appendix:entrop1}) in (\ref{eq_appendix:entrop2}), we have
\begin{equation}
\begin{aligned}
   & \myabs{\mathcal{Y}} - H_{\lambda P_1 + (1-\lambda) P_2}(Y|X) \le \\ & \hspace{5ex} \lambda \big( \myabs{\mathcal{Y}} - H_{P_1}(Y|X) \big) + (1-\lambda)  \big( \myabs{\mathcal{Y}} - H_{P_2}(Y|X) \big)
\end{aligned}
\end{equation}
Finally,
\begin{equation}
    H_{\lambda P_1 + (1-\lambda) P_2}(Y|X) \ge \lambda H_{P_1}(Y|X) + (1-\lambda) H_{P_2}(Y|X).
\end{equation}
\end{proof}

\bibliography{IEEEabrv,main}
\bibliographystyle{IEEEtran}

\end{document}